\documentclass[conference]{IEEEtran}
\IEEEoverridecommandlockouts

\usepackage{cite}
\usepackage{amsmath,amssymb,amsfonts}
\usepackage{algorithmic}
\usepackage{graphicx}
\usepackage{textcomp}
\usepackage{xcolor}
\usepackage{booktabs}
\usepackage{multirow}
\usepackage{subcaption}
\usepackage{url}
\usepackage{enumitem}
\usepackage{hyperref}
\usepackage{cleveref}
\usepackage{balance}
\usepackage[htt]{hyphenat}

\hypersetup{hidelinks}
\def\BibTeX{{\rm B\kern-.05em{\sc i\kern-.025em b}\kern-.08em
    T\kern-.1667em\lower.7ex\hbox{E}\kern-.125emX}}

\newtheorem{proposition}{Proposition}
\newenvironment{proof}{\noindent\textit{Proof.}}{\hfill$\square$\medskip}

\begin{document}

\title{NetKV: Network-Aware Decode Instance Selection\\for Disaggregated LLM Inference}


\author{\IEEEauthorblockN{Mubarak Adetunji Ojewale}
\IEEEauthorblockA{Cloud Competency Centre, National College of Ireland\\
mubarak.ojewale@ncirl.ie}}

\maketitle

\begin{abstract}
Disaggregated LLM inference forces the KV cache to traverse the datacenter network before decoding begins, so transfer time enters directly into the Time to First Token (TTFT) budget. Current schedulers route on compute load and prefix-cache locality alone, ignoring the topological distance and dynamic congestion between prefill and decode instances. We close this gap with a thin operator-to-scheduler interface, the \emph{network cost oracle}, and we prove that ignoring the network term renders cache-aware-only scheduling arbitrarily suboptimal as context length grows. NetKV, the $O(|\mathcal{D}|)$ per-request greedy that consumes this oracle, has tier rankings that are provably robust to stale telemetry. On a 64-GPU four-tier fat-tree simulator driven by Mooncake traces, NetKV reduces mean TTFT by up to 21.2\% over round-robin and 17.6\% over a tuned cache+load-aware scheduler, lifts SLO attainment by up to 20.1 percentage points, and keeps the Time Between Tokens overhead below 0.5\,ms in every condition tested, with no changes to the transport, inference engine, or hardware.
\end{abstract}


\section{Introduction}
\label{sec:intro}

A key requirement for interactive Large Language Model (LLM) serving is the ability to deliver the first generated token within a Service Level Objective (SLO) on the Time to First Token (TTFT), and to sustain a steady inter-token rate, called the Time Between Tokens (TBT), once decoding has begun. Modern production systems for LLM inference have converged on a \emph{disaggregated} architecture in which the compute-intensive prefill phase and the memory-bandwidth-bound decode phase run on separate GPU pools~\cite{distserve,splitwise,mooncake}. This design is now adopted by every major framework, including NVIDIA Dynamo~\cite{dynamo}, llm-d~\cite{llmd}, SGLang~\cite{sglang}, and vLLM~\cite{vllm}, because it eliminates the prefill-decode interference that arises when the two phases share a GPU, and allows the two pools to be sized independently.

Disaggregation, however, introduces a new system bottleneck. The Key-Value (KV) cache generated during prefill must be transferred across the datacenter network to a decode instance before token generation can begin. For a 128K-context Llama-3-70B~\cite{llama3} request, the aggregate KV cache is approximately 40\,GB and, under tensor-parallel (TP) sharding with TP=4, this corresponds to roughly 10\,GB of data crossing each prefill-to-decode GPU pair in parallel. On a 25\,Gbps Remote Direct Memory Access (RDMA) link, that per-pair transfer takes about 3.2\,s, easily dominating the TTFT budget. The transfer time therefore enters directly into the TTFT, and the chosen decode instance determines what fraction of the budget is consumed by data movement rather than computation.

The central scheduling decision in disaggregated inference can thus be stated as follows. Given a request that has completed prefill on instance $p$, which decode instance $d$ should receive its KV cache? Current production schedulers answer this question using two signals: compute load (GPU memory availability, current batch size, queuing delay) and KV-cache locality (whether a candidate decode instance already holds a matching prefix from a prior request). Specifically, DistServe favours same-node placement at deployment time but has no per-request network signal~\cite{distserve}; Mooncake's Conductor scores candidates by cache match and instance load~\cite{mooncake}; llm-d implements a composite scorer over prefix-cache and load within the Kubernetes Gateway API~\cite{llmd}; and Dynamo's KV-aware router directs requests to instances holding matching prefix blocks~\cite{dynamo}. The positioning is summarised in Table~\ref{tab:positioning}. None of these schedulers, however, incorporates a third signal that directly determines transfer latency, namely the \emph{network cost} between the prefill and the decode instance. Mooncake's authors explicitly identify this gap: ``\emph{More difficulty lies in predicting the transfer time because it is determined not only by the size of the transferred data but also by the current network status, especially whether the sending node is under congestion}''~\cite{mooncake}.

This omission reflects a fundamental information asymmetry between two parties. The inference scheduler, operating at the application layer, has no visibility into the physical network topology, link utilisation, or routing decisions. The cluster operator, managing the network fabric, has detailed topology and telemetry but no knowledge of upcoming KV cache transfers, their sizes, or their latency sensitivity. Neither side alone can make a globally optimal placement decision~\cite{Canini2025Cloud}. We postulate that bridging this gap does not require deep integration, but rather a thin, well-defined information-exchange interface that we call a \emph{network cost oracle}. The operator does not need to understand inference semantics, and the scheduler does not need raw topology data. A small, periodically updated signal, specifically a locality-tier classification of instance pairs together with per-tier bandwidth estimates and congestion indicators, is sufficient to make the network bottleneck visible to the scheduler.

We make a threefold contribution. (1)~We formalise decode-instance selection as an online optimisation jointly over compute load, cache locality, and network transfer cost, and we prove that ignoring the network term renders cache-aware-only scheduling arbitrarily suboptimal as context length grows (Proposition~\ref{prop:suboptimal} in \S\ref{sec:formulation}). (2)~We propose NetKV, an $O(|\mathcal{D}|)$ per-request greedy algorithm that consumes the oracle output, and we prove a sufficient condition under which its tier rankings are robust to stale congestion telemetry (Proposition~\ref{prop:staleness} in \S\ref{sec:algorithm}). NetKV is deployable as a scoring plugin to existing schedulers and requires no changes to the transport layer, the inference engine, or the network hardware. (3)~We evaluate NetKV on a 64-GPU four-tier fat-tree cluster simulator driven by Mooncake production traces (\S\ref{sec:eval}). Across seven experiments and five baselines, NetKV reduces mean TTFT by up to 21.2\% over round-robin and 17.6\% over a tuned cache+load-aware scheduler, improves SLO attainment by up to 20.1 percentage points, and keeps the TBT overhead below 0.5\,ms in every condition we tested. The improvements scale with context length, network oversubscription, and background congestion, the regimes in which the network bottleneck of disaggregation is most acute.

\section{Background and Related Work}
\label{sec:related}

\subsection{Disaggregated LLM Inference}

LLM inference comprises two phases with distinct computational profiles. The \emph{prefill} phase processes the entire input prompt in a single forward pass and is compute-bound. The \emph{decode} phase generates output tokens autoregressively, one per step, and is memory-bandwidth-bound. Colocating both phases on the same GPU leads to mutual interference: long prefills block decode batches and inflate the TBT, while decode batches fragment GPU memory that prefill would otherwise use.

DistServe~\cite{distserve} introduced prefill-decode disaggregation and co-optimised GPU allocation and parallelism strategies for each phase independently. Splitwise~\cite{splitwise} proposed a similar phase-splitting architecture with layer-level pipelining that overlaps KV cache transfer with computation. Mooncake~\cite{mooncake} built a production-scale disaggregated platform for Kimi (the Moonshot AI chatbot), featuring a KV-cache-centric distributed store that leverages CPU DRAM, SSD, and RDMA for cache persistence. FlowKV~\cite{flowkv} addresses transfer latency through contiguous memory layout and fused transfer kernels, reporting up to a 96\% reduction in per-transfer latency relative to its own naive baseline. Sarathi-Serve~\cite{sarathi} softens the prefill-decode boundary via chunked prefill but does not eliminate the need to transport KV state across instances. All of these systems acknowledge KV cache transfer as the principal overhead of disaggregation. DistServe mitigates it by preferring same-node placement, Splitwise overlaps transfer with computation, Mooncake optimises the storage and retrieval path, and FlowKV optimises the transfer mechanism itself. None of them, however, makes per-request decode placement decisions as a function of dynamic network state or topological distance, which is the gap that NetKV addresses.

\subsection{KV-Cache-Aware Request Routing}

A parallel line of work optimises request routing based on KV-cache locality. llm-d~\cite{llmd} implements prefix-cache-aware routing within the Kubernetes Gateway API Inference Extension, scoring decode pods by prefix match length. NVIDIA Dynamo~\cite{dynamo} provides a KV-aware router that directs requests to instances holding matching prefix blocks, together with a KV Block Manager for hierarchical cache offloading. AlpaServe~\cite{alpaserve} places model replicas to optimise Service Level Objective (SLO) attainment under bursty workloads. ServerlessLLM~\cite{serverlessllm} introduces locality-aware checkpoint loading. Llumnix~\cite{llumnix} supports live request migration across decode instances for load rebalancing. These systems demonstrate that cache-aware routing produces large TTFT gains over cache-blind baselines. Their per-request routing decisions, however, remain topology-agnostic. A decode pod with a 90\% prefix hit on a congested cross-pod link can yield a worse TTFT than a cold-cache pod on the same rack. NetKV adds the network dimension to this decision, enabling the scheduler to reason about the tradeoff between cache locality and transfer cost.

\subsection{Network-Aware Scheduling for ML Workloads}

Network-aware scheduling has been extensively studied for distributed training. CASSINI~\cite{cassini} models the periodic communication patterns of training jobs as geometric abstractions, enabling interleaving of communication phases across jobs that share network links. vClos~\cite{vclos} jointly optimises topology, routing, communication patterns, and GPU assignments. TopoOpt~\cite{topoopt} co-optimises network topology and parallelisation strategy using reconfigurable optical interconnects. These systems exploit a key property of training, namely that communication is periodic, deterministic, and known in advance: an AllReduce of fixed message size at every iteration. Inference communication is fundamentally different. KV cache transfers are request-driven, stochastic, and variable in size, with the payload proportional to the input sequence length. The communication graph changes with every request, so CASSINI's circle-based abstraction and TopoOpt's offline topology co-optimisation are inapplicable. NetKV addresses this by treating network-aware scheduling as an online decision problem integrated with the per-request scheduling loop.

\subsection{Positioning}

GORGO~\cite{gorgo} studies a related tradeoff between KV cache reuse and network latency at the geo-distributed level, where WAN latencies are measured in hundreds of milliseconds to several seconds. NetKV operates intra-cluster, where the information asymmetry between scheduler and operator is most acute and where topology granularity (rack, pod, spine) directly determines transfer performance. Helix~\cite{helix} formulates LLM placement on heterogeneous GPUs as a max-flow problem on the cluster's compute-network graph, but solves an offline placement problem rather than per-request routing within an already-placed pool, and is therefore complementary to NetKV. Cloud providers such as AWS SageMaker HyperPod, NVIDIA Run:ai, and Kubernetes Topology Aware Scheduling support topology-aware initial placement but not per-request routing within an already-placed pool. Table~\ref{tab:positioning} summarises this positioning. NetKV is, to our knowledge, the first system to incorporate network topology and dynamic congestion awareness into per-request decode-instance selection for disaggregated LLM inference.

\begin{table}[t]
\centering
\caption{Positioning of NetKV relative to existing systems.}
\label{tab:positioning}
\footnotesize
\begin{tabular}{@{}lcccc@{}}
\toprule
System & Cache & Net. & Per-req & Inf. \\
\midrule
DistServe~\cite{distserve}                   & \texttimes & \texttimes$^a$ & \checkmark & \checkmark \\
Mooncake~\cite{mooncake}                     & \checkmark & \texttimes & \checkmark & \checkmark \\
llm-d~\cite{llmd}/Dynamo~\cite{dynamo}       & \checkmark & \texttimes & \checkmark & \checkmark \\
CASSINI~\cite{cassini}/TopoOpt~\cite{topoopt} & N/A       & \checkmark & \texttimes & \texttimes \\
Helix~\cite{helix}                           & \texttimes & \checkmark$^b$ & \texttimes & \checkmark \\
GORGO~\cite{gorgo}                           & \checkmark & \checkmark$^c$ & \checkmark & \checkmark \\
\textbf{NetKV}                               & \checkmark & \checkmark & \checkmark & \checkmark \\
\bottomrule
\multicolumn{5}{@{}l}{\scriptsize $^a$Placement only. $^b$Offline, link-level. $^c$Geo/WAN.}
\end{tabular}
\end{table}

\section{System Model}
\label{sec:system}

\subsection{Cluster Architecture}

We consider a GPU cluster organised as a multi-tier fat-tree~\cite{fattree}, which is the dominant datacenter network architecture for AI workloads. The cluster consists of $N$ GPU instances partitioned into a prefill pool $\mathcal{P}$ and a decode pool $\mathcal{D}$. Instances are organised into \emph{locality domains} determined by the physical topology and indexed by tier. Tier~0 instances share the same physical node and communicate over NVLink or PCIe at 100--600\,GB/s. Tier~1 instances share a rack and communicate over RDMA over Converged Ethernet (RoCE)~\cite{dcqcn} through a Top-of-Rack (ToR) switch at 25--100\,Gbps. Tier~2 instances share a pod and traverse a single aggregation or spine hop at 10--25\,Gbps. Tier~3 instances cross pod boundaries through the core layer at 5--12\,Gbps. Let $\tau(p,d) \in \{0,1,2,3\}$ denote the locality tier of an instance pair $(p,d)$. The static bandwidth for tier $k$ is $B_k$ and the base latency is $L_k$, both known to the operator from hardware specifications and topology.

\subsection{Request and KV Cache Model}

Requests arrive according to a stochastic process. Each request $r$ has input length $\ell_r$ tokens. After prefill on instance $p$, the scheduler selects a decode instance. The KV cache size is deterministic given the model and sequence length:
\begin{equation}
s_r = 2 \cdot n_\text{layers} \cdot n_\text{kv\_heads} \cdot d_\text{head} \cdot \ell_r \cdot b_\text{elem}
\end{equation}
where $n_\text{layers}$, $n_\text{kv\_heads}$, $d_\text{head}$ are model parameters and $b_\text{elem}$ is bytes per element (2 for FP16). For Llama-3-70B~\cite{llama3} (80 layers, 8 KV heads, 128 head dim, Grouped-Query Attention (GQA)~\cite{gqa}), the aggregate KV size is 320\,KB/token; under TP=4, each shard is 80\,KB/token.

We use block-level prefix matching with block size $B_\text{tok}=16$ tokens. Writing $h_r$ for the request's block-hash sequence and $\mathcal{K}_d$ for the cache contents at instance $d$, the cache hit length is $\lambda_r(d) = B_\text{tok} \cdot |\text{LCP}_\text{block}(h_r, \mathcal{K}_d)|$, the number of tokens covered by the longest common block-aligned prefix. The effective transfer size is:
\begin{equation}
s_r^\text{eff}(d) = s_r \cdot \left(1 - \frac{\lambda_r(d)}{\ell_r}\right)
\end{equation}

\subsection{Decode Instance State}

Each decode instance $d \in \mathcal{D}$ maintains local state visible to the scheduler: $m_d$ (available GPU memory), $q_d$ (queued requests), $\beta_d$ (current batch size), and $\mathcal{K}_d$ (set of prefix block hashes in the KV cache). We model the iteration time as $\bar{t}_\text{iter}(\beta) = a + b\beta$, a piecewise-linear function fitted from published profiling data~\cite{distserve}.

\subsection{Network Cost Model}
\label{sec:netcost}

The transfer time from prefill instance $p$ to decode instance $d$ for effective payload $s$ is:
\begin{equation}
T_\text{transfer}(p,d,s) = \frac{s}{B_\text{eff}(p,d)} + L_{\tau(p,d)}
\end{equation}
where effective bandwidth depends on three factors:
\begin{equation}
\label{eq:beff}
B_\text{eff}(p,d) = \frac{B_{\tau(p,d)} \cdot (1 - c_{\tau(p,d)})}{1 + n_\text{inflight}^{\tau}(p)}
\end{equation}

The three factors are: (1)~the \emph{static tier bandwidth} $B_{\tau(p,d)}$, determined by hardware and topology; (2)~the \emph{external congestion factor} $c_{\tau(p,d)} \in [0,1)$, reflecting current utilisation by traffic external to the scheduler's in-flight KV transfers; and (3)~the \emph{self-contention} count $n_\text{inflight}^\tau(p)$, the number of concurrent KV transfers the scheduler has in flight from $p$ on tier $\tau$.

The external and self-contention terms are kept separate to avoid double counting. The operator's telemetry pipeline excludes the scheduler's own KV flows when computing $c_\tau$, which is straightforward when KV transfers use a marked Differentiated Services Code Point (DSCP) class or a dedicated Quality-of-Service (QoS) queue, and $n_\text{inflight}^\tau$ accounts for them directly. When telemetry cannot separate the two, the scheduler sets $n_\text{inflight}^\tau \equiv 0$ and relies on $c_\tau$ alone.

The form of \eqref{eq:beff} composes two standard approximations: $B_\tau(1-c_\tau)$ is the residual-bandwidth approximation used in fluid analyses of TCP and RDMA congestion control~\cite{dcqcn,hpcc} when the network is below saturation, and $1/(1+n_\text{inflight}^\tau)$ is the steady-state share assumed by max-min fairness among the scheduler's in-flight flows plus the new one, the fairness model Data Center Quantized Congestion Notification (DCQCN) converges to over its convergence horizon~\cite{dcqcn}. The product form treats external traffic and the scheduler's own flows as sharing a fair queue independently, which is exact when the two populations are stationary and uncorrelated. The oracle refresh interval $\Delta_\text{oracle} \geq 1$\,s used throughout this work is far longer than the DCQCN convergence horizon, so the steady-state approximation is the right operating regime.

The decision to collapse per-link congestion to per-tier aggregates is a deliberate simplification. It is exact for Tier~0 and Tier~1, where a single ToR uplink carries all cross-rack traffic for a server or rack, and it is approximate for Tier~2 and Tier~3, where Equal-Cost Multi-Path (ECMP) routing spreads flows across multiple spine links and two cross-pod flows can experience different link-level congestion. We discuss the sensitivity to ECMP hash collisions in \S\ref{sec:discussion}.

\emph{Worked example.} Consider a 32K-token retrieval-augmented generation (RAG) request on Llama-3-70B. The aggregate KV cache is $s_r = 320\,\text{KB} \times 32{,}768 \approx 10$\,GB (2.5\,GB per TP=4 shard). The scheduler weighs two candidate decode instances:
\begin{itemize}[leftmargin=*,itemsep=1pt,topsep=2pt]
\item $d_1$, on the same pod as the prefill instance (tier $\tau{=}2$, $B_2{=}50$\,Gbps $\approx 6.25$\,GB/s), with a 50\% prefix hit and one of the scheduler's own KV transfers already in flight on that tier ($n_\text{inflight}^2 = 1$);
\item $d_2$, on a different pod ($\tau{=}3$, $B_3{=}25$\,Gbps $\approx 3.125$\,GB/s), with a much warmer 90\% prefix hit and an idle tier ($n_\text{inflight}^3 = 0$).
\end{itemize}
Under a moderate background congestion of $c_2 = c_3 = 0.2$, the effective transfer sizes are $s_r^\text{eff}(d_1) = 5$\,GB and $s_r^\text{eff}(d_2) = 1$\,GB. Applying~\eqref{eq:beff}:
\begin{align*}
B_\text{eff}(d_1) &= \tfrac{6.25 \cdot 0.8}{1+1} = 2.5\,\text{GB/s}, \\
B_\text{eff}(d_2) &= \tfrac{3.125 \cdot 0.8}{1+0} = 2.5\,\text{GB/s}.
\end{align*}
The transfer times are $T_\text{transfer}(d_1) = 5/2.5 = 2.0$\,s and $T_\text{transfer}(d_2) = 1/2.5 = 0.4$\,s, so the warm-cache cross-pod candidate $d_2$ wins by a factor of five even though it sits on the slower tier. Now perturb the cross-pod congestion alone to $c_3 = 0.5$, keeping every other parameter fixed. The cross-pod effective bandwidth drops to $1.56$\,GB/s, $T_\text{transfer}(d_2)$ rises to $0.64$\,s, and the relative gap collapses from $5\times$ to $3\times$. As soon as queuing differences favour $d_1$ by a few hundred milliseconds, the same-pod candidate wins. This example exhibits all three terms of the cost model: cache locality favours $d_2$, network distance favours $d_1$, and dynamic congestion is the term that flips the verdict.

\subsection{Network Cost Oracle Interface}
\label{sec:oracle}

The oracle is the sole information exchange required between operator and scheduler. The operator sends the scheduler four maps every $\Delta_\text{oracle}$ seconds: a static \texttt{tier\_map} from instance pairs to $\{0,1,2,3\}$, a static \texttt{tier\_bandwidth} map from tier to Gbps, a static \texttt{tier\_latency} map from tier to microseconds, and a dynamic \texttt{congestion} map from tier to $[0,1)$. Optionally, the scheduler may send a per-transfer \texttt{TransferIntent} containing source, destination, payload size, and priority, allowing the operator to anticipate large flows.

The static components derive from existing cluster topology labels such as \texttt{topology.kubernetes.io/zone}  (for Kubernetes) and rack labels, and the static bandwidths are hardware specifications. Only the congestion signal is new, and it is computable from standard switch telemetry (In-band Network Telemetry, sFlow, or SNMP counters) with a lightweight aggregation service. The interface therefore requires no changes to the transport layer (NIXL, NCCL, RDMA), the inference engine (vLLM, SGLang), or the network hardware.

\section{Problem Formulation}
\label{sec:formulation}

\subsection{Objective}

Given request $r$ completing prefill on instance $p$, the scheduler selects a decode instance to minimise the post-prefill latency, which equals TTFT minus the prefill time:
\begin{equation}
\label{eq:objective}
d^* = \arg\min_{d \in \mathcal{D}_r} T_\text{transfer}(p,d,s_r^\text{eff}(d)) + T_\text{queue}(d) + T_\text{decode}(d),
\end{equation}
where the three terms are the network cost, the queuing delay, and the first decode-step time. The feasible set $\mathcal{D}_r = \{d \in \mathcal{D} \mid m_d \geq s_r^\text{eff}(d) + m_\text{min}\}$ is defined by memory constraints, with $m_\text{min}$ a per-instance reserve held back for activations and one additional decode step. This objective directly minimises TTFT, the primary real-time constraint.

The objective treats transfer, queuing, and the first decode step as additive. This corresponds to the execution model in which the decode instance waits for the full KV cache to land before scheduling its first step. Production systems that pipeline transfer with computation (Splitwise, Dynamo, Mooncake) reduce the visible $T_\text{transfer}$ contribution but do not eliminate it under tight memory budgets or long contexts. Our serial model is therefore a conservative upper bound, and internal comparisons remain valid since all baselines are evaluated under the same model.

\subsection{Component Models}

Continuous-batching engines admit new requests at every iteration boundary, so the queuing delay is modelled as
\begin{equation}
T_\text{queue}(d) = \max\!\left(0,\; q_d - (\beta_\text{max} - \beta_d)\right) \cdot \bar{t}_\text{iter}(\beta_d).
\end{equation}
A request joins the active batch immediately if free slots exist, and otherwise it waits one iteration per blocked predecessor. The first decode step, after the request joins the batch on $d$, takes
\begin{equation}
T_\text{decode}(d) = \bar{t}_\text{iter}(\beta_d + 1).
\end{equation}

\subsection{The Three-Way Tradeoff}

The objective encodes a tradeoff that is absent from current schedulers and that involves cache locality, compute load, and the network path. A warm-cache instance on a congested cross-pod link may be slower than a cold-cache same-rack instance, because the cache hit reduces $s_r^\text{eff}$ but the low $B_\text{eff}$ on the cross-pod path more than offsets that reduction. Symmetrically, a lightly loaded distant instance has a low $T_\text{queue}$ but a high $T_\text{transfer}$, while a heavily loaded nearby instance has the opposite profile. The cache-versus-load tradeoff is already present in cache+load-aware (CLA) schedulers (used as our baseline in \S\ref{sec:eval}); NetKV subsumes it and adds the network dimension on top.

\begin{proposition}[Suboptimality of network-oblivious scheduling]
\label{prop:suboptimal}
Consider two tiers with bandwidth ratio $k = B_1/B_3 \geq 1$ and two feasible instances: $d_1$ (same-rack, hit ratio $\rho_1$) and $d_2$ (cross-pod, $\rho_2 \geq \rho_1$). A pure cache-aware scheduler picks $d_2$, but $d_1$ achieves lower TTFT whenever:
\begin{equation}
\begin{aligned}
(1-\rho_1) &< k \cdot \tfrac{1-c_1}{1-c_3} \cdot (1-\rho_2) \\
           &\quad + \tfrac{B_1(1-c_1)}{s_r}\bigl(T_\text{queue}(d_2) - T_\text{queue}(d_1)\bigr).
\end{aligned}
\end{equation}
\end{proposition}

\begin{proof}
$d_1$ achieves lower post-prefill latency than $d_2$ when $s_r(1-\rho_1)/(B_1(1-c_1)) + T_\text{queue}(d_1) < s_r(1-\rho_2)/(B_3(1-c_3)) + T_\text{queue}(d_2)$. Substituting $B_3 = B_1/k$, dividing by $s_r/(B_1(1-c_1))$, and rearranging gives the result.
\end{proof}

\emph{Numerical example.} With $\rho_1{=}0$, $\rho_2{=}0.5$, equal congestion, equal queues, and $k{=}4$ (a realistic same-rack versus cross-pod ratio), the inequality reads $1 < 4 \cdot 0.5 = 2$ and holds. The network-oblivious choice therefore doubles the transfer time. The gap widens with context length, since $s_r$ scales linearly with $\ell_r$ and amplifies the transfer term while $T_\text{queue}$ remains constant, and the gap also widens with congestion asymmetry ($c_3 > c_1$).

\section{Algorithm}
\label{sec:algorithm}

\subsection{NetKV Scheduling Algorithm}

\begin{figure}[t]
\centering
\small
\fbox{\parbox{0.95\columnwidth}{
\textbf{Algorithm 1: NetKV Decode Instance Selection}\\[2pt]
\textbf{Input:} Request $r$, prefill instance $p$, pool $\mathcal{D}$, oracle $O$\\
\textbf{Output:} Selected decode instance $d^*$\\[2pt]
1: $\mathcal{D}_r \gets \{d \in \mathcal{D} \mid m_d \geq s_r^\text{eff}(d) + m_\text{min}\}$\\
2: \textbf{if} $|\mathcal{D}_r| = 0$ \textbf{then reject}($r$)\\
3: \textbf{for each} $d \in \mathcal{D}_r$ \textbf{do}\\
4: \quad $\tau \gets O.\text{tier\_map}(p,d)$\\
5: \quad $B_\text{eff} \gets O.B[\tau] \cdot (1{-}O.c[\tau]) \,/\, (1{+}n_\text{inflight}[\tau][p])$\\
6: \quad $\lambda \gets B_\text{tok} \cdot |\text{LCP}_\text{block}(h_r, \mathcal{K}_d)|$\\
7: \quad $s_\text{eff} \gets s_r \cdot (1 - \lambda/\ell_r)$\\
8: \quad $T_\text{xfer} \gets s_\text{eff}/B_\text{eff} + O.L[\tau]$\\
9: \quad $T_\text{queue} \gets \max(0,\, q_d {-} (\beta_\text{max}{-}\beta_d)) \cdot \bar{t}_\text{iter}(\beta_d)$\\
10: \quad $T_\text{decode} \gets \bar{t}_\text{iter}(\beta_d{+}1)$\\
11: \quad $C[d] \gets T_\text{xfer} + T_\text{queue} + T_\text{decode}$\\
12: \textbf{end for}\\
13: $d^* \gets \arg\min_{d \in \mathcal{D}_r} C[d]$\\
14: $n_\text{inflight}[\tau(p,d^*)][p] \mathrel{+}= 1$\quad // decremented in transfer-done callback\\
15: \textbf{return} $d^*$
}}
\label{alg:netkv}
\end{figure}

Algorithm~1 integrates into the existing scheduling loop. Lines~5--8 add network cost computation, which is a tier lookup, two multiplications, and one division per candidate. The prefix cache lookup on line~6 is already performed by cache-aware schedulers, so NetKV adds negligible overhead on top of the existing scoring path.

\subsection{Complexity}

The algorithm is $O(|\mathcal{D}_r|)$ per decision. In our 64-GPU evaluation with TP=4, the candidate pool after memory filtering satisfies $|\mathcal{D}_r| \leq 12$. The oracle values ($B[\tau]$, $c[\tau]$) are cached locally and refreshed every $\Delta_\text{oracle}$ seconds, with a default of 1\,s, so the per-decision network overhead is effectively zero.

\subsection{Self-Contention Tracking}

The counter $n_\text{inflight}^\tau(p)$ on line~14 prevents the scheduler from over-committing to nearby decode instances when it dispatches multiple requests from the same prefill instance in quick succession. It is incremented on dispatch and decremented by polling the inference engine's existing transfer-complete API: in vLLM, this is \texttt{KVConnectorBase\_V1.get\_finished(finished\_req\_ids)}, which returns the sets of request IDs whose KV-cache send and receive operations have completed and is the same call the engine already uses to release prefill-side buffers; Dynamo exposes an equivalent completion event through its KV router. No new transport-layer notification is required. We cap the counter at a configurable maximum (default 16, roughly the NIC's saturated flow count) to prevent runaway under sustained overload. This signal is available at no cost since the scheduler itself initiated the transfers, and it provides a significant fraction of NetKV's benefit even when the operator does not yet expose dynamic congestion data.

\subsection{Oracle Staleness Analysis}

The congestion signal $c_\tau$ is refreshed every $\Delta_\text{oracle}$ seconds; between refreshes the scheduler reads a stale value $\hat{c}_\tau$.

\begin{proposition}[Staleness tolerance for tier ranking]
\label{prop:staleness}
Let $|\hat{c}_\tau - c_\tau^*| \leq \epsilon$ for all tiers. For candidates on tiers $\tau, \tau'$ with $B_\tau \geq B_{\tau'}$, if $B_\tau(1{-}c_\tau^*) > B_{\tau'}(1{-}c_{\tau'}^*)$, the stale ordering preserves the true ordering when:
\begin{equation}
\label{eq:staleness}
\epsilon < \frac{B_\tau(1{-}c_\tau^*) - B_{\tau'}(1{-}c_{\tau'}^*)}{B_\tau + B_{\tau'}}
\end{equation}
\end{proposition}

\begin{proof}
The worst-case ranking inversion has the stale oracle inflate the slower tier and deflate the faster: $\hat{c}_\tau = c_\tau^* + \epsilon$, $\hat{c}_{\tau'} = c_{\tau'}^* - \epsilon$. Then $\hat{B}_\text{eff}(\tau) = B_\tau(1{-}c_\tau^*) - B_\tau\epsilon$ and $\hat{B}_\text{eff}(\tau') = B_{\tau'}(1{-}c_{\tau'}^*) + B_{\tau'}\epsilon$. Ordering is preserved ($\hat{B}_\text{eff}(\tau) > \hat{B}_\text{eff}(\tau')$) iff $B_\tau(1{-}c_\tau^*) - B_{\tau'}(1{-}c_{\tau'}^*) > (B_\tau + B_{\tau'})\epsilon$, yielding~\eqref{eq:staleness}.
\end{proof}

\emph{Numerical interpretation.} With $B_1/B_3 = 4$ (a 4:1 oversubscribed fat-tree, matching our evaluation) and moderate congestion $c_1^* = c_3^* = 0.3$, the bound becomes $\epsilon < (4 \cdot 0.7 - 1 \cdot 0.7)/(4 + 1) = 0.42$. The oracle can therefore err by up to 42 percentage points before the tier ordering inverts, which is a sizeable tolerance. At the other extreme, if the faster tier is near saturation ($c_\tau^* \to 1$), the right-hand side of (\ref{eq:staleness}) becomes negative and no staleness tolerance exists. This is expected, because at saturation the tier ordering is determined by instantaneous congestion rather than by static bandwidth.

Note that Proposition~\ref{prop:staleness} addresses inversion of the bandwidth ordering only. The queue and cache terms in the full cost $C(d)$ are unaffected by oracle staleness, so the composite cost is at least as well ranked under stale congestion as the bandwidth term alone.

\section{Evaluation}
\label{sec:eval}

\subsection{Experimental Setup}
\label{sec:setup}

We evaluate NetKV \emph{in simulation}. We do not have access to a physical GPU cluster for this work; all results below come from a discrete-event, flow-level simulator described in \S\ref{sec:simulator}, configured to model a 64-H100 cluster. Where we say ``we run Llama-3-70B'' or ``each request's KV transfer is realised as four parallel flows,'' we mean the simulator executes the inference timing model and the network model under those parameters. The validation status of the simulator and the provenance of the timing model are discussed below.

The simulated cluster is a fat-tree of 2 pods, 2 racks per pod, 2 servers per rack, and 8 GPUs per server (64 H100-class GPUs total). The per-tier bandwidths are $B_0{=}450$\,GB/s (NVLink), $B_1{=}100$\,Gbps (ToR), $B_2{=}50$\,Gbps (2:1 oversubscription), and $B_3{=}25$\,Gbps (4:1 oversubscription); base latencies are $L_0{=}1\,\mu$s, $L_1{=}3\,\mu$s, $L_2{=}8\,\mu$s, $L_3{=}15\,\mu$s. The simulator models Llama-3-70B~\cite{llama3} at TP=4 on both prefill and decode, yielding 16 instances in total (4 prefill, 12 decode) with $\beta_\text{max}=64$ and continuous batching~\cite{orca}; each request's KV transfer is modelled as four parallel GPU-to-GPU flows sharing the source NIC. TP=4 is chosen over the more common TP=8 because it yields 16 instances on the 64-GPU cluster (versus 8 at TP=8), which exposes the four-tier topology richly enough for the locality-aware experiments. Llama-3-70B in FP16 fits at TP=4 with comfortable headroom for KV cache and activations (35\,GB of weights per GPU plus 45\,GB of free High Bandwidth Memory (HBM)); Experiment~7 sanity-checks the TP=4 choice by scaling the simulated topology while holding the per-instance configuration fixed.

The inference-timing component of the simulator is the iteration time $\bar{t}_\text{iter}(\beta)$ and prefill latency $T_\text{prefill}(\ell)$. Lacking hardware to profile, we fit these as piecewise-linear functions $\bar{t}_\text{iter}(\beta) = a + b\beta$ and $T_\text{prefill}(\ell) = c\ell + d$ from three published sources: DistServe~\cite{distserve}, vLLM v0.6 benchmarks, and the MLPerf Inference v5.0 datacenter results~\cite{mlperf} (using, in particular, the Juniper Networks 32$\times$H100 Llama-2-70B offline submission as an aggregate sanity check). The fitted aggregate falls within 25\% of the MLPerf number at the matching TP and batch-size configuration. We bias the fit toward faster decode than the published midpoint, so the network term occupies a smaller fraction of TTFT than the optimistic interpretation would suggest; the TTFT improvements we report are therefore a lower bound on what the fitted-conservative regime supports.

The workload is driven by the Mooncake production traces~\cite{mooncake}, comprising 23K real requests with arrival timestamps and input/output lengths. Timestamps are compressed by a single multiplicative factor to achieve the target arrival rate while preserving the burstiness of the original trace. We define three workload profiles. The chatbot profile filters to inputs no larger than 8K tokens with prefix-sharing probability $p_\text{share}=0.3$ and a 2\,s TTFT SLO. The RAG/enterprise profile filters to inputs in $[4\text{K}, 64\text{K}]$ with $p_\text{share}=0.7$ and a 5\,s SLO. The long-context profile filters to inputs above 16K with $p_\text{share}=0.1$ and a 10\,s SLO. Each run simulates 5\,s of warmup and 15\,s of measurement, and every data point is repeated with five independent seeds for variance.

We compare against five baselines: round-robin (RR); load-aware (LA) minimising $T_\text{queue}+T_\text{decode}$; cache-aware (CA) maximising prefix hit length with load as tiebreaker; cache+load-aware (CLA*) with per-workload tuned weights matching the scoring component of Mooncake's Conductor and llm-d's composite scorer; and NetKV-Static, which adds the static tier map and the self-contention counter $n_\text{inflight}^\tau$ to CLA* but withholds the live congestion oracle (labelled \emph{+Self-Contention} in the ablation plots and Table~\ref{tab:ablation}, reflecting the component it adds on top of NetKV-Topo-Only). CLA* weights are tuned by a $10\times10$ grid search in $(w_\text{cache}, w_\text{load}) \in [0.1, 2.0]^2$ at 80\% capacity, on a tuning slice at Mooncake-trace seconds $[0, 30)$ disjoint from the measurement window at $[60, 75)$. Selected weights are $(1.0, 1.0)$ for chatbot and RAG, $(1.5, 0.7)$ for long-context; the NetKV-Full advantage varies by less than 1.5 pp across the $3\times3$ neighbourhood of these points, and an off-design retune at 250\% load preserves the gap within 0.8 pp. Reported metrics: TTFT (mean, P50, P95, P99), TBT (mean, P95), SLO attainment, goodput, per-tier link utilisation, and mean transfer time. TBT is $\bar{t}_\text{iter}(\beta)$ at the moment each request enters the decode batch.

\subsection{Simulator Design and Validation}
\label{sec:simulator}

Our evaluation uses a discrete-event, flow-level simulator that models each request from arrival through prefill, KV transfer, decode, and completion. The fat-tree fabric is captured at flow granularity with per-link max-min fair sharing, the model used by RDMA congestion control~\cite{hpcc}. Each flow's path is determined by ECMP, modelled as uniform random uplink assignment at flow start so correlated flows can collide even below capacity; on every flow arrival or completion, all co-existing flows on shared links are immediately re-evaluated. Background traffic is modelled as a steady-state link utilisation parameterised by a single offered-load fraction, the standard mean-field approximation in fluid analyses of datacenter networks~\cite{dctcp,hpcc}; it is exact for our cost model when the oracle's 1\,s refresh interval exceeds the autocorrelation horizon of the background flows. Experiment~3 sweeps the offered-load fraction in $\{0, 0.05, 0.1, 0.2, 0.4\}$.

The compute model uses the piecewise-linear iteration-time and prefill-latency fits described above. The continuous-batching engine admits new requests at iteration boundaries, so a request arriving mid-iteration waits for the current step to complete before joining the active batch. Each decode instance maintains a running batch with memory tracking, eviction of completed requests, and an LRU-managed KV-block cache. Scheduler decisions are made using only the state that would be available to a real scheduler: per-instance compute metrics refreshed at each scheduling event, and oracle-provided network metrics refreshed every $\Delta_\text{oracle}$ seconds. The scheduler cannot observe per-flow network state or future arrivals.

We validated the network model against analytical predictions: a single flow on an uncontested path matches its tier bandwidth within 0.1\%, $N$ co-existing flows on a bottleneck each receive $1/N$ of capacity, and fair-share reallocation after flow completion converges within one time step. The inference timing model was cross-validated against DistServe~\cite{distserve}, vLLM~\cite{vllm}, and MLPerf~\cite{mlperf} numbers within 25\%. The 5\,s warm-up window ensures cold-cache and empty-queue transients do not contaminate the 15\,s measurement window.

\subsection{Experiment 1: Load Sweep}
\label{sec:exp1}

We sweep the request rate from 50\% to 250\% of the per-workload calibrated capacity on each of the three workload profiles. Table~\ref{tab:load_sweep} shows the RAG profile in detail; the per-workload summary is reported at the end of the subsection.
\begin{table}[t]
\centering
\caption{Load sweep, RAG profile: mean TTFT (mean$\pm$std over five seeds), TBT, SLO attainment, and transfer time at selected rates. Best value per rate in \textbf{bold}. Goodput tracks SLO attainment $\times$ throughput (omitted for space).}
\label{tab:load_sweep}
\footnotesize
\begin{tabular}{@{}llrrrr@{}}
\toprule
Rate & Scheduler & TTFT (ms) & TBT (ms) & SLO & Xfer (ms) \\
\midrule
\multirow{3}{*}{100\%}
& RR         & $1969{\pm}9$  & 12.74 & 0.907 & 993 \\
& CLA*       & $1812{\pm}15$ & 12.71 & 0.923 & 835 \\
& NetKV-Full & $\mathbf{1598{\pm}10}$ & 12.94 & \textbf{0.944} & \textbf{620} \\
\midrule
\multirow{3}{*}{200\%}
& RR         & $2171{\pm}4$  & 13.01 & 0.887 & 1194 \\
& CLA*       & $1995{\pm}25$ & 12.88 & 0.906 & 1018 \\
& NetKV-Full & $\mathbf{1710{\pm}9}$ & 13.29 & \textbf{0.933} & \textbf{733} \\
\midrule
\multirow{3}{*}{250\%}
& RR         & $2224{\pm}4$  & 13.13 & 0.882 & 1246 \\
& CLA*       & $2068{\pm}15$ & 12.97 & 0.899 & 1091 \\
& NetKV-Full & $\mathbf{1793{\pm}9}$ & 13.42 & \textbf{0.925} & \textbf{816} \\
\bottomrule
\end{tabular}
\end{table}
Table~\ref{tab:load_sweep} shows that NetKV-Full achieves the lowest TTFT at every rate, and the gap widens with load. At 100\% load, NetKV-Full reduces mean TTFT by 18.9\% over RR and 11.8\% over CLA*, with SLO attainment lifting 3.7 pp over RR and 2.1 pp over CLA*; mean transfer time drops from 993\,ms (RR) to 620\,ms. At 200\% the gap widens to 21.2\% over RR and 14.3\% over CLA*. The tails follow the means: P99 TTFT at 100\% drops from 12.3\,s (RR) and 11.2\,s (CLA*) to 9.3\,s, and at 250\% from 14.5\,s and 13.6\,s to 11.1\,s. The TBT cost is small and grows slowly: +0.23\,ms over CLA* at 100\% rising to +0.45\,ms at 250\%, well below any practical TBT SLO threshold (50\,ms or more for interactive use).

\emph{Per-workload summary.} The qualitative finding holds across all three profiles. Peak NetKV-Full TTFT reduction over RR is 12.6\% on chatbot (200\% rate, $601{\pm}1$\,ms vs $685{\pm}1$\,ms), 21.2\% on RAG (200\%), and 23.9\% on long-context (75\%, $6175{\pm}26$\,ms vs $8112{\pm}33$\,ms); the corresponding peaks vs CLA* are 6.8\%, 14.3\%, 12.1\%. NetKV-Full wins at every rate in every workload, with the vs-RR gain monotone in context length and seed std under 30\,ms throughout. Goodput tracks SLO: at RAG 200\%, NetKV-Full delivers 11.37\,rps within SLO vs 11.04 (CLA*) and 10.81 (RR).

\subsection{Experiment 2: Context Length Sweep}
\label{sec:exp2}

To isolate the effect of context length, we fix the workload to RAG at 100\% load and parametrically override the per-request input length, sweeping it from 1K to 64K tokens while keeping the arrival timestamps and request count from the underlying Mooncake trace fixed. Four schedulers participate in this sweep: RR, CA, CLA*, and NetKV-Full.


\begin{table}[t]
\centering
\caption{RAG context-length sweep. $\Delta$TTFT is the NetKV-Full mean-TTFT reduction; $\Delta$SLO is the attainment delta in absolute units. Peaks in \textbf{bold}. Seed std on absolute mean TTFT stays below 20\,ms in the schedulable regime (lengths $\leq 16$K) and below 100\,ms in the overload regime (32K--64K).}
\label{tab:context}
\footnotesize
\begin{tabular}{@{}rrrrrr@{}}
\toprule
Length & vs RR & vs CLA* & $\Delta$TBT & vs RR & vs CLA* \\
       & $\Delta$TTFT & $\Delta$TTFT & (ms) & $\Delta$SLO & $\Delta$SLO \\
\midrule
1024  & $-$9.1\%  & $-$2.2\%  & +0.22 & 0.000 & 0.000 \\
4096  & $-$13.3\% & $-$3.2\%  & +0.27 & 0.000 & 0.000 \\
8192  & $-$15.2\% & $-$10.2\% & +0.26 & 0.000 & 0.000 \\
16384 & $-$\textbf{20.2}\% & $-$\textbf{17.6}\% & +0.17 & +\textbf{0.201} & +\textbf{0.136} \\
32768 & $-$6.7\%  & $-$4.3\%  & +0.03 & $-$0.006 & $-$0.006 \\
65536 & $+$0.0\%  & $+$0.2\%  & 0.00  & 0.000 & 0.000 \\
\bottomrule
\end{tabular}
\end{table}

Table~\ref{tab:context} confirms Proposition~\ref{prop:suboptimal}: NetKV's advantage grows with context length while the KV cache fits the per-instance memory budget and the workload remains schedulable. At input length 16K the advantage peaks at 20.2\% mean TTFT reduction over RR and 17.6\% over CLA*, lifting SLO attainment from 0.791 (RR) to 0.992 (a 20.1\,pp improvement). The 16K length is precisely the regime where the transfer term dominates the TTFT budget: the KV cache is large enough (5\,GB aggregate) for the inter-tier bandwidth gap to matter, and small enough for the request to remain within the SLO. Beyond 16K, the workload exceeds the simulated capacity of all schedulers at 100\% load and SLO attainment collapses uniformly; the network choice no longer determines deadline misses. The TBT overhead stays under 0.27\,ms across the sweep and shrinks to zero in the overload regime.

\subsection{Experiment 3: Topology Sensitivity}
\label{sec:exp3}

We fix the request rate at 100\% capacity and sweep two topology axes independently on each of the three workload profiles. The first axis is the cross-pod oversubscription ratio, ranging from 1:1 (no bandwidth gap between tiers) to 8:1 (extreme oversubscription typical of cost-optimised fabrics). The second axis is the background traffic intensity, ranging from 0\% to 40\% of link capacity. Three schedulers participate in the sweep: CLA*, NetKV-Static, and NetKV-Full. The other three baselines, which lack network awareness, would not add information on these axes. Fig.~\ref{fig:topo_allwl} shows the full per-workload sweep.
From Fig.~\ref{fig:topo_allwl}, RAG mean-TTFT advantage of NetKV-Full over CLA* scales along both axes. At 1:1 oversubscription with no background traffic, the inter-tier bandwidth gap is by construction zero and the improvement is only 3.7\%. At 8:1 with 40\% background traffic, the gap widens to 15.9\% ($3644{\pm}45$ vs $3063{\pm}25$\,ms). NetKV's value strengthens precisely where the network bottleneck is most acute. TBT overhead stays within a 0.05\,ms band across the sweep, showing no sensitivity to topology: the TBT cost is constant under network stress while the TTFT benefit scales with it, the asymmetric tradeoff we sought to demonstrate.

\emph{Per-workload summary} (Fig.~\ref{fig:topo_allwl}). The trend holds across all three profiles. At maximum stress (8:1, 40\% bg) NetKV-Full cuts mean TTFT over CLA* by 11.2\% on chatbot, 15.9\% on RAG, and 11.2\% on long-context (SLO lift $+1.5$/$+2.3$/$+6.9$\,pp, with the long-context jump driven by saturation: a small TTFT cut recovers many deadline misses). At minimum stress (1:1, 0\% bg) the gap is 1.3\%/3.7\%/4.0\%, still favourable. Across the full sweep, NetKV-Full beats CLA* in all 60 cells.

\begin{figure*}[]
\centering
\includegraphics[scale=0.45]{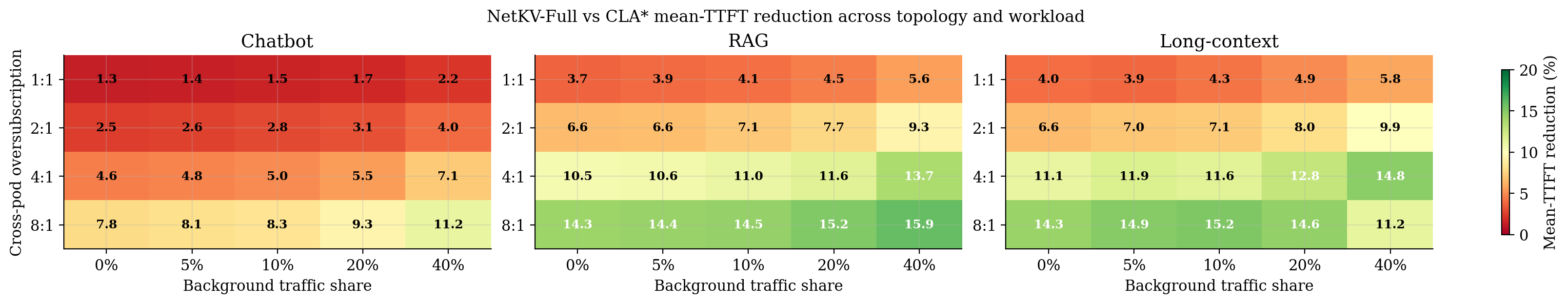}
\caption{NetKV-Full mean-TTFT reduction over CLA* (\%) across the topology sweep for each workload profile. Rows: cross-pod oversubscription ratio. Columns: background-traffic share. NetKV-Full wins in every one of the 60 cells; the reduction grows with both axes and reaches its maximum at the high-stress corner (8:1 oversubscription, 40\% background).}
\label{fig:topo_allwl}
\end{figure*}

\subsection{Experiment 4: Oracle Staleness}
\label{sec:exp4}
We sweep the oracle refresh interval $\Delta_\text{oracle}$ from 100\,ms to 60\,s on the RAG workload at 100\% load, comparing CLA*, NetKV-Static, and NetKV-Full.
\begin{figure*}[]
\centering
\includegraphics[scale=0.4]{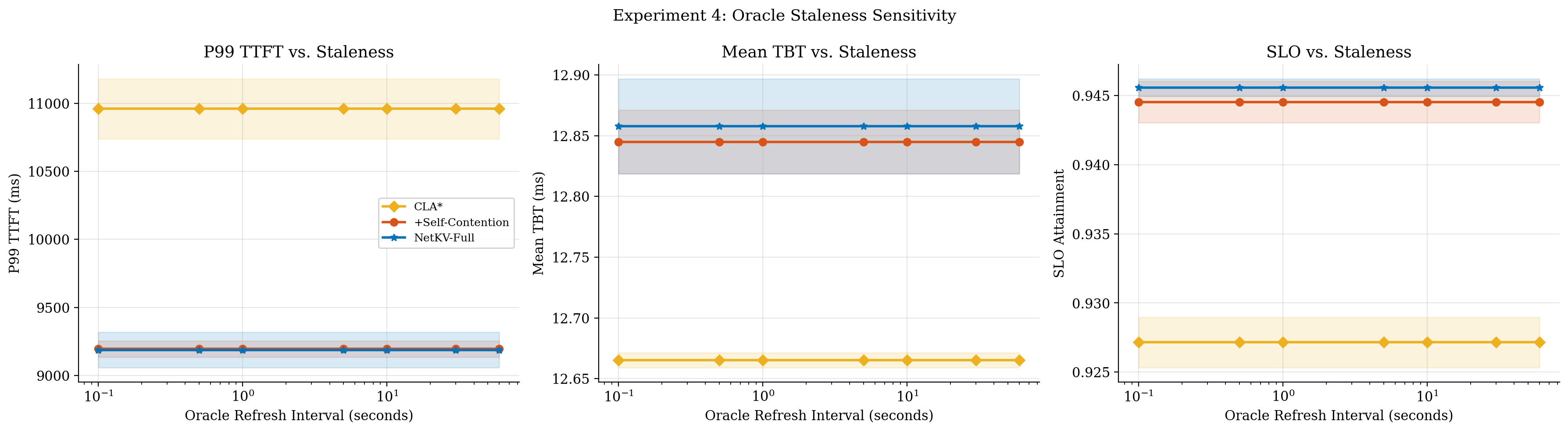}
\caption{Oracle staleness sweep: TTFT, TBT, and SLO are invariant from 100\,ms to 60\,s refresh intervals.}
\label{fig:staleness}
\end{figure*}
All three metrics are essentially invariant across the full range, as Fig.~\ref{fig:staleness} shows NetKV-Full maintains 11.0\% mean TTFT improvement over CLA* (1580\,ms versus 1776\,ms) at every refresh interval, and the SLO attainment delta stays at 1.8 percentage points throughout. Two effects compose to produce this invariance. First, Proposition~\ref{prop:staleness} bounds the worst-case ordering-flip error at 42 percentage points on $c_\tau$ for our 4:1-oversubscribed topology, which is roughly an order of magnitude larger than the measured fluctuation of $c_\tau$ over 60-second windows in our workload, so the tier ordering essentially never inverts. Second, the static tier component dominates the cost (Experiment~6: $>$90\% of the gain from the static tier signal), so even when the dynamic congestion term is stale the static term remains correct. The practical implication is that the oracle can refresh once per minute without degrading scheduling quality. A standard SNMP counter poll at this cadence is sufficient and the telemetry burden on the operator is near zero. The flatness should not be read as evidence that dynamic congestion is unimportant; Experiment~3 shows it contributes meaningfully under the deep-oversubscription, heavy-background-traffic stress regime where the static tier alone leaves a residual gap.

\subsection{Experiment 5: Prefix Sharing Interaction}
\label{sec:exp5}

We sweep the prefix-sharing parameter $p_\text{share}$ from 0.0 to 0.9 on the RAG arrival pattern, comparing CA, CLA*, and NetKV-Full.
\begin{figure*}[t]
\centering
\includegraphics[scale=0.4]{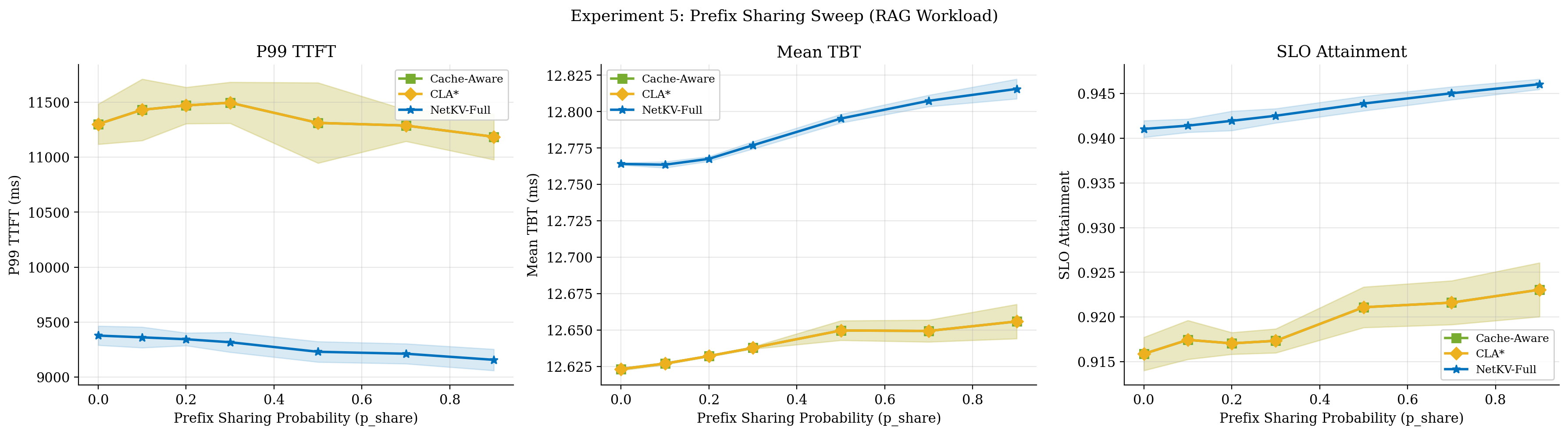}
\caption{Prefix-sharing sweep on the RAG workload: NetKV-Full preserves a roughly constant TTFT advantage over CA and CLA* across the full range, indicating that the network-aware contribution is orthogonal to the cache-aware contribution.}
\label{fig:prefix}
\end{figure*}
Fig.~\ref{fig:prefix} shows the results. Across the full range of $p_\text{share}$ from 0.0 to 0.9, NetKV-Full reduces the mean TTFT by 15.3\% to 15.6\% over both CA and CLA* (since CA and CLA* produce identical decisions when the load component does not break ties), and it lifts the SLO attainment by 2.5 percentage points consistently. The near-constant advantage across the sweep indicates that the network-aware contribution is essentially orthogonal to the cache-aware contribution. The mechanism that underlies the gain is tier-shifting: NetKV routes a large fraction of transfers to Tier~2 (intra-pod, higher bandwidth) rather than the Tier~3 (cross-pod) paths to which CLA* defaults under uniform candidate distribution. The mean transfer time drops from 835\,ms (CLA*) to 620\,ms (NetKV-Full) at the operating point common to this sweep, a 25.8\% reduction.

\subsection{Experiment 6: Ablation Study}
\label{sec:exp6}

To isolate the incremental contribution of each NetKV component, we evaluate a four-step ladder of schedulers on three workload profiles at 100\% load. The ladder begins with CLA*, then adds the static tier map (NetKV-Topo-Only), then the self-contention counter (NetKV-Static), and finally the dynamic congestion signal (NetKV-Full).

\begin{figure*}[t]
\centering
\includegraphics[scale=0.4]{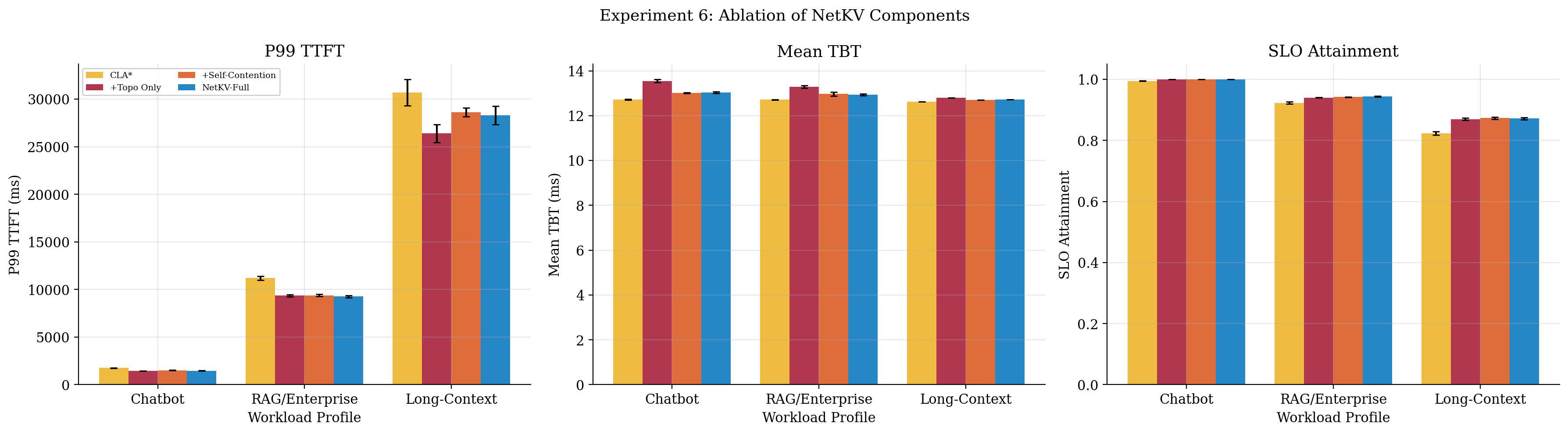}
\caption{Ablation ladder: mean TTFT for CLA*, NetKV-Topo-Only, NetKV-Static, and NetKV-Full across the chatbot, RAG, and long-context workloads.}
\label{fig:ablation}
\end{figure*}

\begin{table}[t]
\centering
\caption{Incremental contribution analysis. TTFT and $\Delta$ are in ms and \%. Seed std per cell: $\leq 7$\,ms (chatbot, 20 seeds), $\leq 15$\,ms (RAG), $\leq 108$\,ms (long-context). All incremental contributions are statistically significant ($p<10^{-4}$, two-sample $t$-test on the chatbot 20-seed run).}
\label{tab:ablation}
\footnotesize
\setlength{\tabcolsep}{3pt}
\begin{tabular}{@{}lrrrrrr@{}}
\toprule
Component & \multicolumn{2}{c}{Chatbot} & \multicolumn{2}{c}{RAG} & \multicolumn{2}{c}{Long-ctx} \\
\cmidrule(lr){2-3} \cmidrule(lr){4-5} \cmidrule(lr){6-7}
 & TTFT & $\Delta$ & TTFT & $\Delta$ & TTFT & $\Delta$ \\
\midrule
CLA*         & 629 & ---      & 1812 & ---      & 7121 & --- \\
+ Static     & 598 & $-$4.9   & 1627 & $-$\textbf{10.2} & 6326 & $-$\textbf{11.2} \\
+ Self-cont. & 596 & $-$0.3   & 1596 & $-$1.9   & 6150 & $-$2.8 \\
+ Dyn.\ cong.& 595 & $-$0.2   & 1592 & $-$0.3   & 6160 & $+$0.2 \\
\midrule
\multicolumn{2}{@{}l}{$\Delta$mean TTFT} & $-$5.4  & & $-$11.9 & & $-$12.1 \\
\multicolumn{2}{@{}l}{$\Delta$P99 TTFT}  & $-$17.3 & & $-$17.5 & & $-$9.4 \\
\multicolumn{2}{@{}l}{$\Delta$SLO (pp)}  & +0.5    & & +2.0    & & +4.8 \\
\multicolumn{2}{@{}l}{$\Delta$TBT (ms)}  & +0.30   & & +0.23   & & +0.10 \\
\bottomrule
\end{tabular}
\end{table}

Table~\ref{tab:ablation} reveals that the static topology signal is the dominant component of NetKV's benefit: adding the static tier map to CLA* delivers 10.2\% (RAG) and 11.2\% (long-context) mean-TTFT reduction, within one percentage point of the full NetKV-Full delta. Self-contention contributes a further 1.9--2.8\% on the harder workloads, and dynamic congestion adds a small residual. The longer-context workloads benefit more because their KV cache, and hence the transfer term, dominates TTFT. The P99 reductions are larger than the mean reductions (e.g., 17.5\% vs 11.9\% on RAG), as the tail benefits most from avoiding cross-pod transfers. The TBT cost from CLA* to NetKV-Full stays below 0.30\,ms because NetKV's routing affects the one-time transfer path while only marginally perturbing the ongoing batch.

The practical implication is that a minimal deployment, in which the operator exposes only the static tier map and no dynamic telemetry, would already capture most of NetKV's benefit. The dynamic congestion signal is the right thing to add when the cluster reaches the network-stress regimes characterised in Experiment~3, but it is not a prerequisite for a useful first deployment.

\subsection{Experiment 7: Scalability}
\label{sec:exp7}

We scale the cluster from 64 GPUs (2 pods) to 1024 GPUs (32 pods) using a flow-level estimator that we cross-validate against the packet-level simulator at 64 and 128 GPUs. The packet-level runs serve as ground truth in the overlap region, while the flow-level estimator carries the trend to the larger scales that are infeasible at the packet level within our wall-clock budget.

\begin{figure*}[t]
\centering
\includegraphics[scale=0.35]{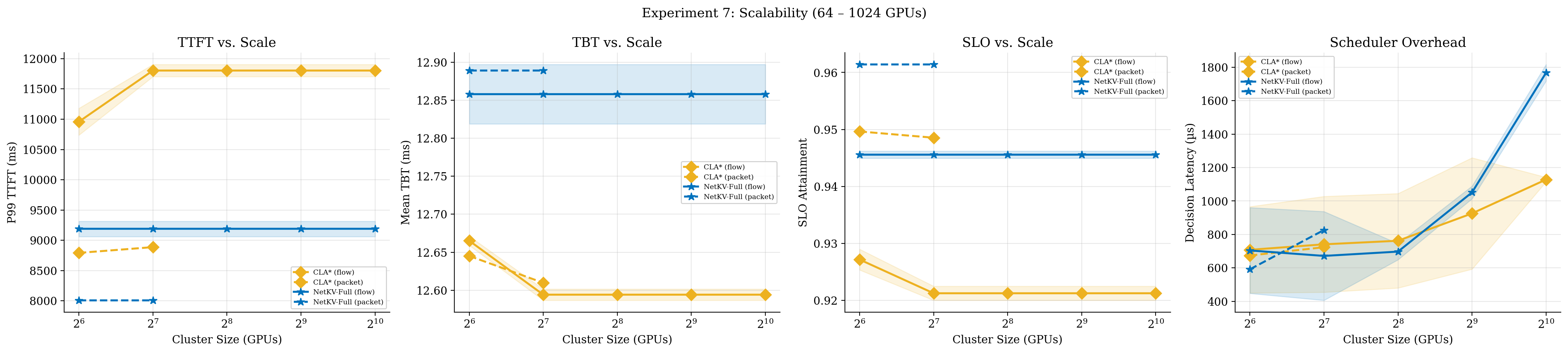}
\caption{Scalability: mean TTFT, mean TBT, SLO attainment, and scheduler decision latency from 64 to 1024 GPUs. The flow-level estimator captures the same relative trend as the packet-level runs in the overlap region.}
\label{fig:scalability}
\end{figure*}

\begin{table}[t]
\centering
\caption{NetKV-Full versus CLA* at increasing cluster scale. Packet rows at 64 and 128 GPUs are ground truth; flow rows at and above 64 are the estimator. ``NK Xfer'' and ``CLA Xfer'' report the mean transfer time (ms) of NetKV-Full and CLA* respectively. $\Delta$SLO is in percentage points. Seed std on absolute mean TTFT is below 16\,ms in every row.}
\label{tab:scale}
\footnotesize
\begin{tabular}{@{}rlrrrr@{}}
\toprule
GPUs & Mode & $\Delta$TTFT & $\Delta$SLO & NK Xfer & CLA Xfer \\
\midrule
64   & packet & $-$7.0\%  & +1.0 & 451 & 559 \\
64   & flow   & $-$11.0\% & +1.8 & 603 & 799 \\
128  & packet & $-$7.4\%  & +1.0 & 451 & 565 \\
128  & flow   & $-$13.6\% & +2.2 & 603 & 851 \\
256  & flow   & $-$13.6\% & +2.2 & 603 & 851 \\
512  & flow   & $-$13.6\% & +1.7 & 603 & 851 \\
1024 & flow   & $-$13.6\% & +1.7 & 603 & 851 \\
\bottomrule
\end{tabular}
\end{table}

Table~\ref{tab:scale} shows three findings. First, in the overlap region between the two simulators (64 and 128 GPUs), the packet-level model reports a 7.0--7.4\% TTFT reduction and the flow-level model reports an 11.0--13.6\% reduction. The two simulators agree qualitatively (NetKV always wins, the gap widens with cluster size) but disagree on the absolute magnitude. Per-condition, the flow-level estimator overestimates the mean TTFT itself by 9--13\% and the NetKV-vs-CLA* gap by 4--6 percentage points across matched configurations; the discrepancy traces to per-flow contention dynamics (DCQCN convergence transients and ECMP hash collisions) that the packet-level model captures and the flow-level model abstracts away. We therefore treat the flow-level numbers at 256, 512, and 1024 GPUs as upper bounds on the true advantage and read the realistic improvement at very large scale as the packet-level 7.4\% number plus a topology-driven trend that we cannot bound without packet-level runs at those scales. Second, the transfer time exhibits a divergent trend across schedulers. CLA*'s transfer time rises from 559\,ms (64 GPUs, packet) to 851\,ms (128 GPUs and above, flow) as larger clusters route a higher fraction of traffic across pod boundaries, while NetKV-Full keeps its transfer time essentially flat at 451\,ms (packet) or 603\,ms (flow). The topology awareness avoids the cross-pod penalty that grows naturally with scale. The flow-level gap is flat at 13.6\% from 128 GPUs upward because the simulator's per-tier abstraction saturates: every prefill--decode pair that would route cross-pod under CLA* already does once the cluster is large enough, so CLA*'s cost stops growing. Real fabrics may exhibit per-link contention effects at very large scale that neither model captures. Third, the scheduler decision latency grows sub-linearly with cluster size, because the candidate pool after memory filtering grows slowly, and it remains below 1.5\,ms even at 1024 GPUs, which is negligible relative to the 600\,ms transfer time it optimises.

\subsection{Cross-Experiment Synthesis}
\label{sec:cross_summary}

Four patterns emerge. (i)~The TTFT gain is driven by transfer-time reduction (19--38\% across conditions). (ii)~The TBT cost is bounded: absolute TBT stays in 12.55--13.42\,ms across all runs, with NetKV-vs-CLA* overhead of 0--0.45\,ms, an order of magnitude below interactive TBT SLOs. (iii)~The gains are orthogonal to cache-aware and load-balanced scheduling (Experiment~5). (iv)~Static topology awareness dominates; dynamic congestion adds a small residual under stress, so a minimal deployment exposing only the static tier map captures most of the benefit.

\section{Discussion}
\label{sec:discussion}

%

\subsection{Deployment Considerations}

NetKV is a scoring plugin rather than a scheduler replacement. In llm-d, the natural integration target is the inference scheduler's pluggable scorer chain~\cite{llmd}: existing scorers expose a \texttt{Score(ctx, pods)$\to$map[Pod]float64} interface in the Endpoint Picker, and NetKV implements one more such scorer alongside the prefix-cache, load, and session-affinity scorers. In Dynamo, the equivalent integration point is the KV-aware router's scoring function. No changes are required to the transport layer (NIXL, NCCL, RDMA), the inference engine (vLLM, SGLang), or the network hardware. The self-contention counter is decremented via the existing transfer-complete event that vLLM's \texttt{KVConnector} and Dynamo's KV router already fire to release the prefill-side KV buffers. The static oracle components derive from existing Kubernetes topology labels and hardware specifications; only the per-tier congestion signal is new, and it is a lightweight aggregation of standard switch telemetry. The Experiment~4 staleness results show that a once-per-minute refresh suffices. When tenants make network-aware decisions, traffic becomes locality-biased and the operator benefits from reduced cross-pod bandwidth demand, as the tier-distribution analysis in the next subsection makes explicit.


\subsection{Tier-Shifting Mechanism}

The mechanism underlying NetKV's TTFT improvement is a redistribution of KV transfers toward lower-latency tiers, which we call \emph{tier shifting} (Table~\ref{tab:tier_dist}, RAG at 100\% load). CLA* distributes transfers between Tier~2 and Tier~3 in roughly 32:68 proportion since it treats the two tiers as equivalent on its compute and cache criteria. NetKV-Full reverses this to 69:31 by scoring same-pod candidates higher than cross-pod ones under matched compute and cache state. The cross-pod fabric is the slowest path in our topology (4:1 oversubscription), so each cross-pod transfer carries a multi-second penalty that NetKV avoids whenever a same-pod candidate exists. This shift cuts the mean transfer time from 835 to 620\,ms, a 25.7\% reduction that translates directly into the 11.8\% TTFT improvement in Experiment~1.

\begin{table}[t]
\centering
\caption{Fraction of transfers traversing each tier under CLA* and NetKV-Full. RAG workload at 100\% load. Tier~0 and Tier~1 are unreached in this configuration because the instance placement does not co-locate prefill with decode at server or rack granularity.}
\label{tab:tier_dist}
\small
\begin{tabular}{@{}lrr@{}}
\toprule
Tier & CLA* & NetKV-Full \\
\midrule
Tier 0 (same-node)  & 0.0\%   & 0.0\%   \\
Tier 1 (same-rack)  & 0.0\%   & 0.0\%   \\
Tier 2 (same-pod)   & 32.0\%  & 68.9\%  \\
Tier 3 (cross-pod)  & 68.0\%  & 31.1\%  \\
\midrule
Mean transfer time  & 835\,ms & 620\,ms \\
\bottomrule
\end{tabular}
\end{table}

This tier-shifting mechanism explains why NetKV's benefit scales with the oversubscription ratio in Experiment~3. Higher oversubscription widens the per-tier bandwidth gap, which makes the penalty for cross-pod placement more severe and the reward for avoiding it correspondingly larger. The mechanism also explains why the gain is largely orthogonal to prefix sharing (Experiment~5): the cache hit reduces $s_r^{\text{eff}}$ on either tier, but it does not change the relative bandwidth available on that tier, so the tier choice continues to dominate the transfer-time budget.

\subsection{Limitations}

Our evaluation has some limitations that we report in the following. First, all numbers come from a flow-level simulator. A real-cluster validation is the most important next step, since the simulator faithfully models ECMP, fair sharing, and multi-tier bandwidth contention but does not capture NIC-level effects (driver queueing, completion-queue contention) or kernel-bypass overheads; we plan a real cluster testbed measurement to calibrate the simulator's $B_\text{eff}$ against measured DCQCN behaviour. Second, the $\bar{t}_\text{iter}(\beta)$ model is fitted from published data rather than measured on our target configuration, and although we bias the fit conservatively, a single profile on the deployment hardware would replace the triangulation. Third, the serial transfer-queue-decode model overestimates the visible $T_\text{transfer}$ in production systems that pipeline transfer with computation; our comparative advantage persists since all baselines use the same model, but the absolute gap would shrink. Fourth, the oracle collapses per-link congestion to per-tier aggregates, which is exact for Tier~0 and Tier~1 but approximate for Tier~2 and Tier~3 under ECMP. The cost model therefore cannot see ECMP hash collisions that the simulator does generate, which may understate the P99 advantage of NetKV in heavy-tailed contention regimes; per-link telemetry would close this gap. Fifth, the CLA* baseline matches the \emph{scoring} component of Mooncake's Conductor and llm-d's composite scorer, but Mooncake additionally performs KV-cache replication across tiers that we do not model; a full comparison against the deployed Conductor is left for future work. Sixth, the per-request greedy does not jointly optimise across concurrent requests; a batch-level formulation could yield better results at higher computational cost. Seventh, we assume an FP16 KV cache. INT8 or block-quantised KV~\cite{cachegen} roughly halves $s_r$ and proportionally reduces the network term's weight, so NetKV remains beneficial but the gains shrink. The TP=4 evaluation choice also leaves open how the gains transfer to the more common TP=8 production configuration, where the candidate pool is sparser; a single TP=8 data point is part of our planned validation work.

\subsection{Future Work}

Four directions follow naturally. The most direct is a real-cluster validation on a testbed with GPUs and RDMA-capable NICs, comparing per-transfer completion times under controlled background load against the simulator's predictions. Multi-hop KV routing extends NetKV to architectures that stage KV state through intermediate caches in CPU DRAM or SSDs~\cite{mooncake,lmcache}: the oracle exposes tier information for both hops and the cost model sums the two transfer times, with the greedy generalising naturally. Predictive congestion replaces the instantaneous snapshot with a lightweight time-series prediction (exponential smoothing or ARIMA) over historical telemetry, leveraging the large staleness tolerance of Proposition~\ref{prop:staleness}. Mixture-of-experts inference makes the TBT itself network-sensitive (expert-parallel All-to-All every step), so NetKV would extend to a joint TTFT/TBT optimisation. Finally, multi-tenant operation on shared infrastructure raises a price-of-anarchy question that we leave open; a basic deployment can surface tenant-specific oracle views (the same tier map but with congestion telemetry restricted to the tenant's allocated share) so each tenant sees a self-consistent picture, with global fairness delegated to the operator's QoS layer.

\section{Conclusion}
\label{sec:conclusion}

Disaggregated LLM inference makes the KV-cache transfer the dominant TTFT determinant, yet current schedulers ignore network state. NetKV closes the gap with a lightweight cost oracle between scheduler and operator, with no changes to the transport, the inference engine, or hardware. On a 64-GPU fat-tree, NetKV cuts mean TTFT by up to 21.2\% over RR and 17.6\% over CLA*, lifts SLO attainment by up to 20.1 percentage points, and keeps TBT overhead below 0.5\,ms throughout. The ablation isolates the static tier signal as the dominant contributor: a minimal deployment that exposes only the tier map captures most of the gain.



\bibliographystyle{IEEEtran}
\bibliography{refs}

@inproceedings{distserve,
  author = {Zhong, Yinmin and Liu, Shengyu and Chen, Junda and Hu, Jianbo and Zhu, Yibo and Liu, Xuanzhe and Jin, Xin and Zhang, Hao}, title = {DistServe: disaggregating prefill and decoding for goodput-optimized large language model serving}, year = {2024}, isbn = {978-1-939133-40-3}, publisher = {USENIX Association}, address = {USA}, booktitle = {Proceedings of the 18th USENIX Conference on Operating Systems Design and Implementation}, articleno = {11}, numpages = {18}, location = {Santa Clara, CA, USA}, series = {OSDI'24}
}

@article{splitwise,
  author    = {Pratyush Patel and Esha Choukse and Chaojie Zhang and Aashaka Shah and {\'I}{\~n}igo Goiri and Saeed Maleki and Ricardo Bianchini},
  title     = {{Splitwise}: Efficient Generative {LLM} Inference Using Phase Splitting},
  booktitle = {Proceedings of the 51st Annual International Symposium on Computer Architecture (ISCA)},
  year      = {2024},
}

@inproceedings{mooncake,
  author = {Qin, Ruoyu and Li, Zheming and He, Weiran and Cui, Jialei and Tang, Heyi and Ren, Feng and Ma, Teng and Cai, Shangming and Zhang, Yineng and Zhang, Mingxing and Wu, Yongwei and Zheng, Weimin and Xu, Xinran}, title = {Mooncake: A KVCache-centric Disaggregated Architecture for LLM Serving}, year = {2025}, publisher = {Association for Computing Machinery}, address = {New York, NY, USA}, issn = {1553-3077}, url = {https://doi.org/10.1145/3773772}, doi = {10.1145/3773772}, journal = {ACM Trans. Storage}, month = nov
}

@inproceedings{vllm,
  author = {Kwon, Woosuk and Li, Zhuohan and Zhuang, Siyuan and Sheng, Ying and Zheng, Lianmin and Yu, Cody Hao and Gonzalez, Joseph and Zhang, Hao and Stoica, Ion}, title = {Efficient Memory Management for Large Language Model Serving with PagedAttention}, year = {2023}, isbn = {9798400702297}, publisher = {Association for Computing Machinery}, address = {New York, NY, USA}, url = {https://doi.org/10.1145/3600006.3613165}, doi = {10.1145/3600006.3613165}, booktitle = {Proceedings of the 29th Symposium on Operating Systems Principles}, pages = {611–626}, numpages = {16}, location = {Koblenz, Germany}, series = {SOSP '23}
}

@misc{dynamo,
  author       = {{NVIDIA}},
  title        = {{Dynamo}: A Datacenter-Scale Distributed Inference Framework},
  howpublished = {Open-source project, \url{https://github.com/ai-dynamo/dynamo}},
  year         = {2025},
}

@misc{llmd,
  author       = {{llm-d project}},
  title        = {{llm-d}: Kubernetes-native Distributed Inferencing},
  howpublished = {CNCF Sandbox, \url{https://llm-d.ai}},
  year         = {2025},
}

@inproceedings{cassini,
  author = {Rajasekaran, Sudarsanan and Ghobadi, Manya and Akella, Aditya}, title = {CASSINI: network-aware job scheduling in machine learning clusters}, year = {2024}, isbn = {978-1-939133-39-7}, publisher = {USENIX Association}, address = {USA}, booktitle = {Proceedings of the 21st USENIX Symposium on Networked Systems Design and Implementation}, articleno = {78}, numpages = {18}, location = {Santa Clara, CA, USA}, series = {NSDI'24}
}

@inproceedings{topoopt,
  title={$\{$TopoOpt$\}$: Co-optimizing network topology and parallelization strategy for distributed training jobs},
  author={Wang, Weiyang and Khazraee, Moein and Zhong, Zhizhen and Ghobadi, Manya and Jia, Zhihao and Mudigere, Dheevatsa and Zhang, Ying and Kewitsch, Anthony},
  booktitle={20th USENIX Symposium on Networked Systems Design and Implementation (NSDI 23)},
  pages={739--767},
  year={2023}
}

@misc{gorgo,
  title        = {{GORGO}: Maximizing {KV}-Cache Reuse While Minimizing Network Latency in Cross-Region {LLM} Load Balancing},
  howpublished = {arXiv:2602.11688},
  year         = {2026},
  month        = feb,
}

@inproceedings{helix,
  author = {Mei, Yixuan and Zhuang, Yonghao and Miao, Xupeng and Yang, Juncheng and Jia, Zhihao and Vinayak, Rashmi}, title = {Helix: Serving Large Language Models over Heterogeneous GPUs and Network via Max-Flow}, year = {2025}, isbn = {9798400706981}, publisher = {Association for Computing Machinery}, address = {New York, NY, USA}, url = {https://doi.org/10.1145/3669940.3707215}, doi = {10.1145/3669940.3707215}, booktitle = {Proceedings of the 30th ACM International Conference on Architectural Support for Programming Languages and Operating Systems, Volume 1}, pages = {586–602}, numpages = {17}, location = {Rotterdam, Netherlands}, series = {ASPLOS '25}
}

@misc{flowkv,
  author       = {Weiqing Li and Guochao Jiang and Xiangyong Ding and Zhangcheng Tao and Chuzhan Hao and Chenfeng Xu and Yuewei Zhang and Hao Wang},
  title        = {{FlowKV}: A Disaggregated Inference Framework with Low-Latency {KV} Cache Transfer and Load-Aware Scheduling},
  howpublished = {arXiv:2504.03775},
  year         = {2025},
}

@inproceedings{sarathi,
  author    = {Amey Agrawal and Nitin Kedia and Ashish Panwar and Jayashree Mohan and Nipun Kwatra and Bhargav S. Gulavani and Alexey Tumanov and Ramachandran Ramjee},
  title     = {{Sarathi-Serve}: Taming Throughput--Latency Tradeoff in {LLM} Inference with Chunked Prefill},
  booktitle = {Proceedings of the 18th USENIX Symposium on Operating Systems Design and Implementation (OSDI)},
  year      = {2024},
}

@inproceedings{alpaserve,
  author    = {Zhuohan Li and Lianmin Zheng and Yinmin Zhong and Vincent Liu and Ying Sheng and Xin Jin and Yanping Huang and Zhifeng Chen and Hao Zhang and Joseph E. Gonzalez and Ion Stoica},
  title     = {{AlpaServe}: Statistical Multiplexing with Model Parallelism for Deep Learning Serving},
  booktitle = {Proceedings of the 17th USENIX Symposium on Operating Systems Design and Implementation (OSDI)},
  year      = {2023},
}

@inproceedings{serverlessllm,
  author    = {Yao Fu and Leyang Xue and Yeqi Huang and Andrei-Octavian Brabete and Dmitrii Ustiugov and Yuvraj Patel and Luo Mai},
  title     = {{ServerlessLLM}: Locality-Enhanced Serverless Inference for Large Language Models},
  booktitle = {Proceedings of the 18th USENIX Symposium on Operating Systems Design and Implementation (OSDI)},
  year      = {2024},
}

@inproceedings{llumnix,
  author    = {Biao Sun and Ziming Huang and Hanyu Zhao and Wencong Xiao and Xinyi Zhang and Yong Li and Wei Lin},
  title     = {{Llumnix}: Dynamic Scheduling for Large Language Model Serving},
  booktitle = {Proceedings of the 18th USENIX Symposium on Operating Systems Design and Implementation (OSDI)},
  year      = {2024},
}

@inproceedings{cachegen,
  author    = {Yuhan Liu and Hanchen Li and Yihua Cheng and Siddhant Ray and Yuyang Huang and Qizheng Zhang and Kuntai Du and Jiayi Yao and Shan Lu and Ganesh Ananthanarayanan and Michael Maire and Henry Hoffmann and Ari Holtzman and Junchen Jiang},
  title     = {{CacheGen}: {KV} Cache Compression and Streaming for Fast Large Language Model Serving},
  booktitle = {Proceedings of the ACM SIGCOMM 2024 Conference},
  year      = {2024},
}

@article{vclos,
  author = {Han, Xinchi and Zhao, Shizhen and Lv, Yongxi and Cao, Peirui and Jiang, Weihao and Yang, Qinwei and Liu, Yunzhuo and Lin, Shengkai and Jiang, Bo and Liu, Ximeng and Cui, Yong and Zhou, Chenghu and Wang, Xinbing}, title = {vClos: Network contention aware scheduling for distributed machine learning tasks in multi-tenant GPU clusters}, year = {2025}, issue_date = {Aug 2025}, publisher = {Elsevier North-Holland, Inc.}, address = {USA}, volume = {268}, number = {C}, issn = {1389-1286}, url = {https://doi.org/10.1016/j.comnet.2025.111285}, doi = {10.1016/j.comnet.2025.111285}, journal = {Comput. Netw.}, month = aug, numpages = {13}
}

@inproceedings{hpcc,
  author    = {Yuliang Li and Rui Miao and Hongqiang Harry Liu and Yan Zhuang and Fei Feng and Lingbo Tang and Zheng Cao and Ming Zhang and Frank Kelly and Mohammad Alizadeh and Minlan Yu},
  title     = {{HPCC}: High Precision Congestion Control},
  booktitle = {Proceedings of the ACM SIGCOMM 2019 Conference},
  year      = {2019},
}

@misc{mlperf,
  author       = {{MLCommons}},
  title        = {{MLPerf} Inference: Datacenter v5.0 Results},
  howpublished = {\url{https://mlcommons.org/benchmarks/inference-datacenter/}},
  year         = {2025},
  month        = apr,
  note         = {Llama-2-70B offline, e.g., Juniper Networks 32$\times$H100 submission at 82,749 tokens/s.},
}

@misc{lmcache,
  author       = {Yuhan Liu and Jiayi Yao and Hanchen Li and Yihua Cheng and others},
  title        = {{LMCache}: An Efficient {KV} Cache Layer for Enterprise-Scale {LLM} Inference},
  howpublished = {arXiv:2510.09665},
  year         = {2025},
  month        = oct,
}

@inproceedings{orca,
  title={Orca: A distributed serving system for $\{$Transformer-Based$\}$ generative models},
  author={Yu, Gyeong-In and Jeong, Joo Seong and Kim, Geon-Woo and Kim, Soojeong and Chun, Byung-Gon},
  booktitle={16th USENIX symposium on operating systems design and implementation (OSDI 22)},
  pages={521--538},
  year={2022}
}

@inproceedings{sglang,
  author = {Zheng, Lianmin and Yin, Liangsheng and Xie, Zhiqiang and Sun, Chuyue and Huang, Jeff and Yu, Cody Hao and Cao, Shiyi and Kozyrakis, Christos and Stoica, Ion and Gonzalez, Joseph E. and Barrett, Clark and Sheng, Ying}, title = {SGLang: efficient execution of structured language model programs}, year = {2024}, isbn = {9798331314385}, publisher = {Curran Associates Inc.}, address = {Red Hook, NY, USA},  booktitle = {Proceedings of the 38th International Conference on Neural Information Processing Systems}, articleno = {2000}, numpages = {27}, location = {Vancouver, Canada}
}

@inproceedings{fattree,
  author = {Al-Fares, Mohammad and Loukissas, Alexander and Vahdat, Amin}, title = {A scalable, commodity data center network architecture}, year = {2008}, issue_date = {October 2008}, publisher = {Association for Computing Machinery}, address = {New York, NY, USA}, volume = {38}, number = {4}, issn = {0146-4833}, url = {https://doi.org/10.1145/1402946.1402967}, doi = {10.1145/1402946.1402967}, journal = {SIGCOMM Comput. Commun. Rev.}, month = aug, pages = {63–74}, numpages = {12} }

@inproceedings{gqa,
  title={{GQA}: Training generalized multi-query transformer models from multi-head checkpoints},
  author={Ainslie, Joshua and Lee-Thorp, James and De Jong, Michiel and Zemlyanskiy, Yury and Lebr{\'o}n, Federico and Sanghai, Sumit},
  booktitle={Proceedings of the 2023 Conference on Empirical Methods in Natural Language Processing},
  pages={4895--4901},
  year={2023}
}

@misc{llama3,
  title={The llama 3 herd of models},
  author={Grattafiori, Aaron and Dubey, Abhimanyu and Jauhri, Abhinav and Pandey, Abhinav and Kadian, Abhishek and Al-Dahle, Ahmad and Letman, Aiesha and Mathur, Akhil and Schelten, Alan and Vaughan, Alex and others},
  journal={arXiv preprint arXiv:2407.21783},
  year={2024}
}

@inproceedings{dcqcn,
  author = {Zhu, Yibo and Eran, Haggai and Firestone, Daniel and Guo, Chuanxiong and Lipshteyn, Marina and Liron, Yehonatan and Padhye, Jitendra and Raindel, Shachar and Yahia, Mohamad Haj and Zhang, Ming}, title = {Congestion Control for Large-Scale RDMA Deployments}, year = {2015}, issue_date = {October 2015}, publisher = {Association for Computing Machinery}, address = {New York, NY, USA}, volume = {45}, number = {4}, issn = {0146-4833}, url = {https://doi.org/10.1145/2829988.2787484}, doi = {10.1145/2829988.2787484}, journal = {SIGCOMM Comput. Commun. Rev.}, month = aug, pages = {523–536}, numpages = {14} }

@inproceedings{dctcp,
 author = {Alizadeh, Mohammad and Greenberg, Albert and Maltz, David A. and Padhye, Jitendra and Patel, Parveen and Prabhakar, Balaji and Sengupta, Sudipta and Sridharan, Murari}, title = {Data center TCP (DCTCP)}, year = {2010}, isbn = {9781450302012}, publisher = {Association for Computing Machinery}, address = {New York, NY, USA}, url = {https://doi.org/10.1145/1851182.1851192}, doi = {10.1145/1851182.1851192}, booktitle = {Proceedings of the ACM SIGCOMM 2010 Conference}, pages = {63–74}, numpages = {12}, location = {New Delhi, India}, series = {SIGCOMM '10}
}

@inproceedings{Canini2025Cloud,
author = {Canini, Marco and Benson, Theophilus A. and Bianchini, Ricardo and Goiri, \'{I}\~{n}igo and Kosti\'{c}, Dejan and Pietzuch, Peter and Peter, Simon},
title = {Cloud abstractions for AI workloads},
year = {2025},
isbn = {9798400715723},
publisher = {Association for Computing Machinery},
address = {New York, NY, USA},
url = {https://doi.org/10.1145/3725783.3764395},
doi = {10.1145/3725783.3764395},
booktitle = {Proceedings of the 16th ACM SIGOPS Asia-Pacific Workshop on Systems},
pages = {98–105},
numpages = {8},
location = {Lotte Hotel World, Emerald Hall, Seoul, Republic of Korea},
series = {APSys '25}
}

\balance

\end{document}